\documentclass{lms}

\usepackage{amsmath} 
\usepackage{amsfonts}
\usepackage{amssymb}

\newtheorem{theorem}{Theorem}[section]
\newtheorem{lemma}[theorem]{Lemma}
\newtheorem{proposition}[theorem]{Proposition}
\newtheorem{corollary}[theorem]{Corollary}
\newtheorem{example}[theorem]{Example}
\renewcommand{\vec}[1]{\mathbf{#1}}

\DeclareMathOperator{\tr}{tr}

\newcommand{\beq}{\begin{equation}}  
\newcommand{\eeq}{\end{equation}}  
\newcommand{\bear}{\begin{array}}
\newcommand{\eear}{\end{array}}

\newcommand{\fib}{\mathfrak{f}} 
\newcommand\lax{{\bf L}}
\newcommand\mma{{\bf M}}
\newcommand\kax{{\bf K}}
\newcommand\cax{{\bf C}}
\newcommand\shi{{\cal S}}
\newcommand\rma{{\mathrm{a}}}

\newnumbered{assertion}{Assertion}    
\newnumbered{conjecture}{Conjecture}  
\newnumbered{definition}{Definition}
\newnumbered{hypothesis}{Hypothesis}
\newnumbered{remark}{Remark}
\newnumbered{note}{Note}
\newnumbered{observation}{Observation}
\newnumbered{problem}{Problem}
\newnumbered{question}{Question}
\newnumbered{algorithm}{Algorithm}

\title[Linearisable recurrences with the Laurent property]
{A family of linearisable recurrences with the Laurent property} 

\author{A.N.W. Hone  and C. Ward} 

\classno{11B37 (primary), 13F60, 37J35 (secondary)} 

\extraline{CW was supported by EPSRC studentship EP/P50421X/1.}

\begin{document}
\maketitle
\begin{abstract}
We consider a family of nonlinear recurrences with the Laurent property. 
Although these recurrences are not generated by mutations in a cluster algebra, 
they fit within the broader framework of Laurent phenomenon algebras, as introduced 
recently by Lam and Pylyavskyy.  Furthermore, each member 
of this family is shown to be linearisable in two different ways, in the sense that 
its iterates satisfy both a linear relation with constant coefficients 
and a linear relation with periodic coefficients. Associated monodromy matrices 
and first integrals are constructed, and the  
connection with the dressing chain for Schr\"odinger operators is also explained.
\end{abstract}
\section{Introduction}

\noindent 
A nonlinear recurrence relation is said to possess the Laurent property if all 
of its iterates are Laurent polynomials in the initial data with integer coefficients. 
Particular recurrences 
of the form
\begin{equation}\label{genform}
x_{n+N}x_n = F \left( x_{n+1},\ldots,x_{n+N-1} \right), 
\end{equation} 
with $F$ being a polynomial, first gained wider notice through the article by  Gale \cite{gale}, 
which highlighted the fact that such nonlinear relations can unexpectedly generate integer sequences. 
The second order recurrence 
\beq \label{a1} 
x_{n+2}x_n = x_{n+1}^2 +1
\eeq 
is one of the simplest examples of this type.  
With the initial values $x_0=1$, $x_1=1$, this generates a sequence beginning 
$1,1,2,5,13,34, 89,233, \ldots$, and it turns out that $x_n$ is an integer for all $n$, 
but it is not obvious why this should be so. The fact that (\ref{a1}) has the 
Laurent property provides an explanation: if $x_0$ and $x_1$ are considered as 
variables, then each $x_n$ is an element of the Laurent polynomial ring 
$\mathbb{Z}[x_0^{\pm 1}, x_1^{\pm 1}]$, and so generates an integer 
when evaluated  at $x_0=x_1=1$. 

One of the motivations behind 
Fomin and Zelevinksky's cluster algebras, introduced in \cite{ca1}, 
was to provide an axiomatic framework for the Laurent property, which 
arises in  many different areas of mathematics. A cluster is an $N$-tuple ${\bf x}= (x_1,x_2, \ldots , x_N)$ 
which can be mutated in direction $j$ for each choice of $j\in\{1,\ldots ,N\}$, to produce a new cluster ${\bf x}'$ 
with components $x_k'=x_k$ for $k\neq j$ and $x_j'$ determined by the exchange relation 
\beq\label{exch} 
x_j'x_j = \prod_{b_{ij}>0}x_{i}^{b_{ij}} +\prod_{b_{ij}<0}x_{i}^{-b_{ij}}  
\eeq 
for a coefficient-free (geometric type) cluster algebra, 
where $B=(b_{ij})$ is  
an associated skew-symmetrizable $N\times N$ integer matrix. There is also 
an associated operation of matrix mutation, $B\to B'$, the precise details of which are 
not needed here.       

In general, sequences of mutations in a cluster algebra do not generate orbits of a 
fixed map, since the exponents in (\ref{exch}) both depend on $j$ and vary under the 
mutation of  the matrix $B$. However, Fordy and Marsh \cite{fordymarsh09} found criteria 
for a skew-symmetric matrix $B$ (or the equivalent quiver) to vary periodically under 
a cyclic sequence of mutations, and in particular showed how this can correspond 
to the iteration of a single recurrence relation of the form (\ref{genform}), with 
$F$ being a sum of two monomials obtained from an exchange relation (\ref{exch}).  
For the classification of mutation-periodic quivers in  \cite{fordymarsh09},  the building-blocks 
were provided by the affine $A$-type quivers (referred to there as ``primitives''), and  
it was also shown that sequences of cluster variables generated from these affine quivers 
satisfy linear relations with constant coefficients. The same result was found independently in \cite{assem}, 
where it was conjectured (and also proved for type $D$) that there should be linear relations for 
cluster variables associated with all the affine Dynkin types; the full conjecture was proved in  \cite{keller}. 

In \cite{fordy10} it was shown that  the $\tilde{A}_{1,N-1}$ 
exchange relations produce integrable maps for $N$ even;  
in the non-commutative setting, 
the exchange relations for $N$ odd were considered in \cite{difrancesco}. 
The simplest case ($N=2$) is just the recurrence (\ref{a1}), which preserves the Poisson bracket $\{ x_n,x_{n+1}\}=x_nx_{n+1}$, 
and has a conserved quantity $C$  (independent of $n$) that appears as a non-trivial coefficient 
in the linear relation for the sequence of cluster variables $(x_n)$, that is 
$$x_{n+2}-Cx_{n+1}+x_n=0, \quad \mathrm{with} \quad 
C=\frac{x_n}{x_{n+1}} + \frac{x_{n+1}}{x_n} + \frac{1}{x_nx_{n+1}}. $$ 
Whenever a nonlinear recurrence is such that its iterates satisfy a  linear relation, as is the case here, we say that it is {\it linearisable}.

Q-systems, which arise from the Bethe ansatz for quantum integrable models, 
provide further examples of nonlinear recurrences  which are obtained 
from sequences of cluster mutations \cite{kedem}, and are also linearisable in the above sense 
\cite{difrancesco}. They correspond to characters of representations of Yangian algebras, 
as well as being reductions of discrete Hirota equations \cite{nahmkeegan}.  The simplest case is the 
$A_1$  Q-system, which coincides with    (\ref{a1}). 
In  addition to the aforementioned examples coming from affine $A$-type quivers, yet another family  
of linearisable recurrences obtained from cluster mutations is considered  in  \cite{fordyhone12}.

In this paper we will consider the family of nonlinear recurrences given by 
\begin{equation}\label{myrec}
x_{n+N}x_n=x_{n+N-1}x_{n+1}+a\sum\limits_{i=1}^{N-1} x_{n+i}, 
\end{equation}
where $N\geq 2$ is an integer and $a$ is a constant parameter. We will show that all of these recurrences have the Laurent property, 
and they are linearisable. However, observe that the right hand side of (\ref{myrec}) has $N$ terms, 
which means that for $N\geq 3$ it cannot be obtained from an exchange relation in a cluster algebra, 
since for that to be so the polynomial $F$ should have the same binomial form as (\ref{exch});   
the case $N=2$ is not an exchange relation either (but see Example 4.11 in \cite{fordyhone12}).

The inspiration for (\ref{myrec}) comes from results of 
Heideman and Hogan \cite{heidhogan08}, who considered odd order recurrences of the form 
\begin{equation}\label{hhodd}
x_{n+2k+1}x_n = x_{n+2k}x_{n+1}+x_{n+k}+x_{n+k+1},
\end{equation}
and showed that an integer sequence is generated for initial data $x_0=x_1=\ldots=x_{2k}=1$. Their argument is  based on the fact that each  
of these integer sequences also satisfies  a linear recurrence relation, which can be   restated thus: 
\begin{proposition}[({\cite[Lemma 2]{heidhogan08}})] 
\label{hhthm}
For each integer $k\geq 1$, the iterates of recurrence (\ref{hhodd})  
with the initial data $x_i=1$, $ i=0,\ldots,2k$,  
satisfy the homogeneous linear relation
\begin{equation}
 x_{n+6k}-(2k^2+8k+4)\big(x_{n+4k}-x_{n+2k}\big)-x_n=0 \label{heholin} 
\qquad \forall n\in{\mathbb Z}. 
\end{equation}
\end{proposition} 
When $N=3$ the recurrence (\ref{myrec}) coincides with (\ref{hhodd}) for $k=1$, and is 
a special case of a more general third order 
recurrence that is shown to be linearisable in \cite{honechapt}. 

In the case of (\ref{myrec}) it turns out that the integer sequence generated by choosing 
all $N$ data to take the value 1 is related to the Fibonacci numbers, which we denote by 
$\fib_n$ (with the convention that $\fib_0=0$, $\fib_1=1$). The analogue of Proposition \ref{hhthm} 
is as follows.  
\begin{proposition}\label{inits}  
For each integer $N\geq 2$, the iterates of recurrence (\ref{myrec})  
with the initial data $x_i=1$, $ i=0,\ldots,p=N-1$,  
satisfy the homogeneous  linear relation 
\beq \label{inik}
x_{n+3p}-(\fib_{2p+2}+p\fib_{2p}-\fib_{2p-2}+1) \big(x_{n+2p}-x_{n+p}\big)-x_n=0 
\qquad \forall n\in{\mathbb Z}.
\eeq 
\end{proposition} 
\subsection{Outline of the paper}  
Although (\ref{myrec}) cannot be obtained from a cluster algebra, in the next section we briefly explain how it fits into the broader framework of Laurent phenomenon algebras (LP algebras), introduced 
in \cite{lam12}, which immediately shows that the Laurent property holds. In the third section we present Theorem \ref{thmlin}, our main result on the linearisability 
of (\ref{myrec});  
in addition to a linear relation with coefficients that are first integrals (i.e. conserved quantities like $C$ above), the iterates also satisfy a linear relation with periodic coefficients.   
Although the result essentially follows by adapting methods from  \cite{fordyhone12}, the complete proof involves the detailed structure of monodromy matrices 
associated with an integrable Hamiltonian system, namely the periodic dressing chain  of \cite{vesshab93}; this connection is explained in section 4. 
Section 5 contains our conclusions, while the properties of the particular family of integer sequences  in Proposition \ref{inits} are reserved for an appendix.  

\section{Laurent phenomenon algebras}\label{lpalg}
A seed in a cluster algebra is a pair $({\bf x}, B)$, where ${\bf x}$ is a cluster and $B$ is a skew-symmetrizable integer matrix;  
for each index $j$, an associated exchange polynomial is given by the right hand side of (\ref{exch}). 
In the setting of LP algebras, introduced recently by Lam and Pylyavskyy \cite{lam12},  the exchange polynomials are allowed to be  irreducible polynomials with an arbitrary number of terms, rather than just binomials.
The Laurent phenomenon for certain recurrences of the type (\ref{genform})
with non-binomial $F$ was already observed in Gale's article \cite{gale}, with  Fomin and Zelevinsky's  Caterpillar Lemma eventually providing a 
means to prove the Laurent property in this more general setting \cite{fz02}. The axioms of an LP algebra, which we briefly describe below, guarantee 
that mutations of an initial seed always produce Laurent polynomials in the initial cluster variables.

A seed in an LP algebra of rank $N$ is a pair of $N$-tuples $t=(\textbf{x},\textbf{F})$, where ${\bf x}$ is a cluster, and the entries of $\textbf{F}$ 
are exchange polynomials $F_j\in  {\cal P}:= S[x_1,\ldots,x_N]$, where $S$ is a  coefficient ring. The exchange polynomials must satisfy 
the conditions that (1)  $F_j$ is irreducible in ${\cal P}$ and not divisible by any variable $x_i$, and (2) $F_j$ does not involve the variable $x_j$, for $j=1,\ldots,N$. 
A  set of exchange Laurent polynomials $\{\hat{F}_1,\ldots,\hat{F}_N\}$ is also defined such that, for all  $j$,  
(i)  there are $\rma_i  \in \mathbb{Z}_{\leq 0}$ for each $i\neq j$ with $\hat{F}_j = F_j \prod_{i\neq j}x_i^{\rma_i}$, and 
(ii) $\hat{F}_i|_{x_j\leftarrow F_j/x} \in S[x_1^{\pm 1},\ldots,x_{j-1}^{\pm 1},x^{\pm 1}, x_{j+1}^{\pm 1},\ldots,  x_N^{\pm 1}]$ and is not divisible by $F_j$.
Mutation in direction $j$ gives $(\textbf{x},\textbf{F})\rightarrow(\textbf{x}',\textbf{F}')$, where the cluster variables of the new seed are given by $x_k'=x_k$ for $k\neq j$ and $x_j'=\hat{F}_j /x_j$.  
For the new exchange polynomials, set 
$$G_k =F_k \Big|_{x_j \leftarrow \frac{\hat{F}_j |_{x_k\leftarrow0}}{x_j'}} $$ 
and define $H_k$ by removing all common factors with $\hat{F}_j|_{x_k\leftarrow 0}$ from $G_k$; then $F_k'=MH_k$ where $M$ is a Laurent monomial in the new cluster variables. Lam and Pylyavskyy proved in \cite[Theorem~5.1]{lam12} that the Laurent phenomenon holds for LP algebras, which means that cluster variables obtained by arbitrary sequences of mutations belong to $S[x_1^{\pm 1}, \ldots ,x_N^{\pm 1}]$. 

To see how each  of  the recurrences (\ref{myrec}) fits into the LP algebra framework, we take the coefficient ring $S=\mathbb{Z}[a]$. 
The initial seed is then given by
$
t_1 = \left\{ (x_1,F_1),\ldots, (x_N,F_N)\right\}
$, 
with the irreducible exchange polynomials 
$$ F_1= x_Nx_2+a\sum\limits_{i=2}^N x_i, \quad F_N   =  x_{N-1}x_1+a\sum\limits_{i=1}^{N-1} x_i, \quad  
F_k=  a+x_{k-1}+x_{k+1}\quad \mathrm{otherwise}.
$$ 
Mutating $t_1$  in direction 1 gives a new seed $t_2=\left\{ (x_{N+1},F_1'),(x_2,F_2'),\ldots (x_N,F_N')\right\}$,  
containing one new cluster variable $x_1'=x_{N+1}$, 
where
$$ 
F_1' = F_1, \quad F_2' = x_{N+1}x_3+a\sum\limits_{i=3}^{N+1}x_i,  \quad  F_N' = a+x_{N-1}+x_{N+1}, \quad  F_k'  =  F_k \quad \mathrm{otherwise}.
$$ 
The seeds $t_1$ and $t_2$ are similar to each other, in the sense of 3.4  in  \cite{lam12} (see also section 7  therein): they are transformed one to the other 
by $(x_1,x_2,\ldots,x_N) \to (x_2,x_3,\ldots, x_{N+1})$, corresponding to a single iteration of (\ref{myrec}), and successive iterations are generated by 
mutating $t_2$ in direction 2, and so forth.  
For the rest of the paper, it will be more convenient to take  $x_0,x_1,\ldots, x_{N-1}$ as initial data, and introduce the ring of Laurent polynomials 
$$ 
{\cal R}_N:= \mathbb{Z}[a,x_0^{\pm 1},\ldots,x_{N-1}^{\pm 1}], 
$$ 
in order to make the following statement. 
\begin{theorem}\label{lpropthm}
For each $N$ the Laurent property holds for (\ref{myrec}), i.e.\ $x_n \in {\cal R}_N  
\quad \forall n\in\mathbb{Z}$.
\end{theorem}
Working in the ambient field of fractions $\mathbb{Q}(a,x_0,\ldots ,x_{N-1})$, 
it is also worth noting  that, from the form  of (\ref{myrec}), the iterates $x_n$ can be written 
as subtraction-free rational expressions in $a$ and the initial data, and hence are all non-vanishing.  
%
\section{Linearisability}

In order to show that the iterates of 
(\ref{myrec}) satisfy linear relations, we begin by noting 
that the  nonlinear recurrence can be rewritten using a $2\times 2$ determinant as
\begin{equation}\label{2det} 
D_n=
a\sum\limits_{i=1}^{N-1} x_{n+i} , \qquad 
\mathrm{where} \quad 
D_n :=   
\left|
\begin{matrix}
x_n & x_{n+N-1} \\
x_{n+1} & x_{n+N}
\end{matrix}
\right| 
,
\end{equation}
and proceed to consider the $3\times3$ matrix
\begin{equation}\label{psi3}
\Psi_n = 
\begin{pmatrix} 
x_n & x_{n+N-1} & x_{n+2N-2} \\
x_{n+1} & x_{n+N} & x_{n+2N-1} \\
x_{n+2} & x_{n+N+1} & x_{n+2N}
\end{pmatrix}.
\end{equation} 
Iterating (\ref{myrec}) with initial data $(x_0,x_1, \ldots ,x_{N-1})$ 
is equivalent to iterating the birational map 
\beq\label{bir}
\varphi: \qquad 
(x_0,x_1, \ldots ,x_{N-1}) \mapsto 
\left( x_1,x_2,\ldots , 
\Big( 
x_{N-1} x_1 +a\sum_{i=1}^{N-1}
\Big) 
/x_0 \right) .    
\eeq 
The key to obtaining linear relations for the variables $x_n$ is the observation 
that the determinant of the matrix $\Psi_n$ is invariant under $n\to n+1$, meaning 
that it provides a conserved quantity for the map $\varphi$. 
\begin{lemma}\label{psi_inv}
The determinant of the matrix $\Psi_n$ is invariant under the map $\varphi$ defined by the recurrence (\ref{myrec}). 
In other words, if it is written 
in terms of the initial data 
as 
\begin{equation}\label{mu} 
\det \Psi_n=\mu(x_0,x_1,\ldots,x_{N-1}), 
\end{equation}
then the function $\mu$ is a first integral, i.e. it satisfies $\varphi^* \mu = \mu$. 
\end{lemma}
\begin{proof}
To begin with, we note that the identity 
\beq\label{did} 
D_{n+1}-D_n = a(x_{n+N}-x_{n+1})  
\eeq 
follows from (\ref{2det}).  
We wish to show that $\det \Psi_{n+1} = \det \Psi_{n}$ for all $n$. 
However, since all formulae 
remain valid under shifting each of the indices by an arbitrary amount, 
it is sufficient to perform the calculation with $n=0$ for ease of notation.  
Using the method of Dodgson condensation \cite{dodgson} to  expand 
the $3\times 3$ determinant in terms of $2 \times 2$ minors, we calculate 
\begin{eqnarray*}
\det \Psi_0  &= & \frac{1}{x_{N}} 
\left| 
\begin{array}{cc} 
 D_0 & D_{N-1} \\
D_{1} & D_{N} 
\end{array} 
\right| 
=
\frac{1}{x_{N}} 
\left| 
\begin{array}{cc} 
 D_1 +a(x_1-x_N) & D_{N}+a(x_N-x_{2N-1})  \\
D_{1} & D_{N} 
\end{array} 
\right| 
\\    
&= &  
\frac{a}{x_{N}}\Big(  
D_N(x_1-x_N)-D_1(x_N-x_{2N-1})
\Big) , 
\end{eqnarray*}
where we have used (\ref{did}). 
Similarly, by using (\ref{did}) again, we have 
\begin{eqnarray*} 
\det \Psi_1  &= & \frac{1}{x_{N+1}} 
\left| 
\begin{array}{cc} 
 D_1 & D_{N} \\
D_{2} & D_{N+1} 
\end{array} 
\right|  
=
\frac{1}{x_{N}} 
\left| 
\begin{array}{cc} 
 D_1 & D_{N} \\
D_{1} +a(x_{N+1}-x_2)  & D_{N} +a(x_{2N}-x_{N+1}) 
\end{array} 
\right| 
\\    
&= &  
\frac{a}{x_{N}}\Big(  
D_1(x_{2N}-x_{N+1}) -D_N(x_{N+1}-x_2) 
\Big)
. 
\end{eqnarray*}
Taking the difference and   multiplying by a common denominator 
 gives 
\begin{eqnarray*}
\frac{x_N x_{N+1}} {a} \left( \det\Psi_1-\det\Psi_0 \right)& = & 
D_N\Big( -x_N (x_{N+1}-x_{2}) -x_{N+1}(  x_{1}-x_{N} ) 
\Big)   \\ 
&& 
+ D_1\Big( x_N (x_{2N}-x_{N+1}) +x_{N+1}( x_{N}-  x_{2N-1}) 
\Big)   \\ 
& = &  -D_N D_1   + D_1D_N =0, 
\end{eqnarray*} 
as required. 
Hence the determinant of the $3 \times 3$ matrix $\Psi_n$ is a conserved quantity (independent of $n$). 
Starting from the matrix (\ref{psi3}) with $n=0$, $\det \Psi_0$ 
can be rewritten as a rational function of the initial values $x_0, \ldots ,x_{N-1}$, 
denoted $\mu$  as in (\ref{mu}),  
by repeatedly using (\ref{myrec}) to express $x_{2N}$ as a rational function of 
terms of lower index, then $x_{2N-1}$, etc. By construction, the pullback 
of this function satisfies $\varphi^* \mu =\mu \cdot \varphi = \mu$,   
so $\mu$ is a first integral for the map (\ref{bir}).  
\end{proof}


\begin{remark}\label{munon} 
The first integral $\mu$ is a rational function, or more precisely a Laurent polynomial, 
in terms of the variables $x_0,\ldots ,x_{N-1}$. To verify that it is not identically zero, 
it is enough to check that one specific choice of initial values gives a non-vanishing 
value of $\mu$. In particular, for the sequence beginning with all $N$ initial values equal to 1, 
the value of $\mu$ can be calculated by using the formulae (\ref{terms}) found in 
the Appendix, which gives 
$\mu (1,1,\ldots ,1) =\fib_{2p}p^2\neq 0$ (where $p=N-1$).    
\end{remark} 

\begin{corollary} \label{cor4det} 
If the sequence $(x_n)$ satisfies (\ref{myrec}), 
then $\det\hat{\Psi}_n =0$, where 
\begin{equation}\label{psi4}
\hat{\Psi}_n = 
\begin{pmatrix}
x_n & x_{n+N-1} & x_{n+2N-2} & x_{n+3N-3} \\
x_{n+1} & x_{n+N} & x_{n+2N-1} & x_{n+3N-2}\\
x_{n+2} & x_{n+N+1} & x_{n+2N} & x_{n+3N-1} \\
x_{n+3} & x_{n+N+2} & x_{n+2N+1} & x_{n+3N}
\end{pmatrix}
.
\end{equation}
\end{corollary}
\begin{proof} 
Note that from (\ref{2det}) it follows that $D_n$ is a subtraction-free rational 
expression in the initial data, hence is non-vanishing for all $n$. 
Applying Dodgson condensation again 
yields
\begin{eqnarray*}
\det\hat{\Psi}_n
& = & \frac{\det\Psi_n\det\Psi_{n+N} - \det\Psi_{n+1}\det\Psi_{n+N-1}} 
{D_{n+N}} =0, 
\end{eqnarray*}
since,  by Lemma \ref{psi_inv},  $\det \Psi_n$ is independent of $n$.  
\end{proof} 
The fact that the $4\times4$ matrix in (\ref{psi4}) 
has a non-trivial kernel is enough to show that sequences  
$(x_n)$ generated by (\ref{myrec}) also satisfy linear recurrence relations.  
There are two types of relation, corresponding to the kernel of 
$\hat{\Psi}_n$, and to the kernel of its transpose, respectively. 
In order to describe these relations more explicitly, 
it is helpful to introduce the 
functions 
\begin{equation}\label{Jfn}
J_n = \frac{x_{n+2}+x_n+a}{x_{n+1}}, 
\end{equation}
which turn out to be periodic with respect to $n$. 
%
%
\begin{theorem}\label{thmlin}
The iterates of the recurrence (\ref{myrec}) satisfy the pair of homogeneous linear relations 
\begin{equation}\label{Klin}
x_{n+3N-3}-K(x_{n+2N-2}-x_{n+N-1})-x_n = 0, 
\end{equation}
\begin{equation}\label{Jlinear}
x_{n+3}-(1+J_{n+1})x_{n+2}+(1+J_n)x_{n+1}-x_n = 0, 
\end{equation}
where $K$ is a first integral (independent of $n$), and 
the coefficients of (\ref{Jlinear}) are periodic, such that  $J_{n+N-1}=J_n$ for all $n$. 
\end{theorem}
\begin{proof} 
From Corollary \ref{cor4det}, the existence of the two linear relations follows by arguments which are almost identical 
to those used for the proof of Lemma 5.1 in  \cite{fordyhone12}. We first sketch the essential argument for (\ref{Klin}), 
before discussing further details 
of the coefficients in each relation.  

Without loss of generality, a non-zero vector ${\bf w}_n$  in the kernel of $\hat{\Psi}_n$ can be written 
in normalised form as $\textbf{w}_n=(A_n,B_n,C_n,-1)^T$. Indeed, the last entry cannot 
vanish, for otherwise there would be a non-trivial kernel for the $3\times 3$ matrix $\Psi_n$ given by (\ref{psi3}),  
contradicting the remark that $\mu$  is not identically zero.  
The first three rows of the equation $\hat{\Psi}_n\textbf{w}_n=0$ produce the linear system
\begin{equation}\label{3ker}
\Psi_n 
\begin{pmatrix}
A_n \\
B_n \\
C_n
\end{pmatrix}
=
\begin{pmatrix}
x_{n+3N-3} \\
x_{n+3N-2} \\
x_{n+3N-1}
\end{pmatrix}, 
\end{equation}
while the last three rows of the same equation give 
\begin{equation}\label{3kernext}
\Psi_{n+1} \begin{pmatrix}
A_n \\
B_n \\
C_n
\end{pmatrix}
=
\begin{pmatrix}
x_{n+3N-2}\\
x_{n+3N-1} \\
x_{n+3N}
\end{pmatrix}.
\end{equation} 
Each of the  systems (\ref{3ker}) and (\ref{3kernext}) 
can be solved for the vector $(A_n,B_n,C_n)^T$, and the 
two different expressions that result differ by a shift $n\to n+1$, 
which implies that 
this vector 
must be independent of $n$. 
Upon applying Cramer's rule to the system (\ref{3ker}), the first component can be simplified as  
$A_n =  \det\Psi_{n+N-1}/\det\Psi_n = 1 $, 
by Lemma \ref{psi_inv}. 
Supposing that they are not trivially constant, the other two components are 
first integrals, so we may set $K^{(1)}=C_n$ and $K^{(2)}=-B_n$, 
and then the equation for ${\bf w}_n$ gives the relation 
\begin{equation}\label{Klinear}
x_{n+3N-3}-K^{(1)}x_{n+2N-2}+K^{(2)}x_{n+N-1}-x_n =0,
\end{equation}
where the coefficients are all independent of $n$. 

By considering the kernel of $\hat{\Psi}_n^T$, an almost identical argument shows that 
$(x_n)$ also satisfies a four-term linear relation with coefficients that are periodic with period $N-1$. 
However, this four-term relation can be obtained  more directly in another way.  
Upon setting 
$ 
E_n:=x_{n+N}x_n-x_{n+N-1}x_{n+1}-a (x_{n+1}+\ldots +x_{n+N-1})$,  
which vanishes for all $n$ whenever (\ref{myrec}) holds, 
taking the forward difference $E_{n+1}-E_n=0$ and rearranging yields
\begin{equation}\label{Jeqn}
\frac{x_{n+N+1}+x_{n+N-1}+a}{x_{n+N}} = \frac{x_{n+2}+x_n+a}{x_{n+1}}. 
\end{equation}
Comparison of the above identity with (\ref{Jfn}) reveals that the sequence of quantities $J_n$ 
is periodic with period $p=N-1$: the left hand side of (\ref{Jeqn})  is $J_{n+p}$, and the right hand side is $J_n$. 
Now rearranging 
(\ref{Jfn}) 
clearly gives an inhomogeneous linear equation, namely 
\begin{equation}\label{Jarrange}
x_{n+2}-J_nx_{n+1}+x_n +a = 0.
\end{equation}
The forward difference operator 
applied 
to the latter equation 
produces 
(\ref{Jlinear}). 
\end{proof} 
As we shall see, the coefficients $K$ and $J_n$ are themselves Laurent polynomials in the initial data, 
meaning that the each of the linear relations (\ref{Klinear}) and (\ref{Jlinear}) offers an alternative proof of Theorem \ref{lpropthm}. 
The periodicity of the quantities (\ref{Jfn}) immediately implies the following. 
\begin{corollary} 
For all $n\in\mathbb{Z}$ the iterates of (\ref{myrec}) satisfy $x_n\in\mathbb{Z}[a,x_0,x_1,J_0,J_1,\ldots , J_{N-2}]$.  
\end{corollary}
Note that the latter result gives a stronger 
version of the Laurent property for the nonlinear recurrence (\ref{myrec}).  Indeed, it is clear from the formula (\ref{Jfn}) that $J_n\in{\cal R}_N$ for $n=0,\ldots, N-3$, 
while (using (\ref{myrec}) to substitute for $x_N$) the same formula  for $n=N-2$  gives 
$$
J_{N-2} = \frac{x_{N-2}}{x_{N-1}} + \frac{x_1}{x_0} + \frac{a}{x_{N-1}x_0}\sum\limits_{i=0}^{N-1} x_{i} 
\in {\cal R}_N . 
$$

Observe that the foregoing proof of Theorem \ref{thmlin} 
is not yet complete: so far we have not shown that the coefficients  $K^{(1)}$ and $K^{(2)}$ in (\ref{Klinear}) coincide, 
which is required for the relation (\ref{Klin}) to hold. 
Cramer's rule applied to the linear system (\ref{3ker}) produces formulae for these coefficients as ratios of $3\times 3$ determinants, 
but it is not apparent from these formulae why  $K^{(1)} = K^{(2)}=K$.  An explanation for this coincidence is postponed until the next section, 
but first we present an example.  

\begin{example}\label{N3} 
For $N=3$ the nonlinear recurrence (\ref{myrec}) is 
\begin{equation}\label{3rdmyrec}
x_{n+3}x_n = x_{n+2}x_{n+1} +a (x_{n+2}+x_{n+1}),
\end{equation}
which is the same as 
Heideman and Hogan's recurrence (\ref{hhodd}) for $k=1$ 
and $a=1$, and is also a special case of the  third order recurrence 
considered in section 5 of  
\cite{honechapt}. 
In accordance with (\ref{Klinear}), the iterates of (\ref{3rdmyrec}) satisfy 
$
x_{n+6}-K(x_{n+4}-x_{n+2}) - x_n =0$, 
where the first integral $K$ can be expressed as a Laurent polynomial in the initial data $x_0,x_1,x_2$ as follows:
\beq \label{constK}
K  
=  
1+ \frac{x_0}{x_2} +\frac{x_2}{x_0} + a\left(\frac{x_0}{x_1x_2} + \frac{x_2}{x_0x_1}+  
2\sum_{i=0}^2\frac{1}{x_i} 
\right) 
+ 
a^2 \left(\frac{1}{x_0x_1}+\frac{1}{x_1x_2}+\frac{1}{x_2x_0}\right) 
.
\eeq 
This agrees with the result of Theorem 5.1 in \cite{honechapt}, 
and with the particular integer sequence considered in 
\cite{heidhogan08}: 
putting the initial data $x_0=x_1=x_2=1$ with $a=1$ into (\ref{constK}) gives the constant value $K=14$, in accordance with 
Proposition \ref{hhthm} for $k=1$. 
 
The iterates of (\ref{3rdmyrec}) also satisfy the homogeneous relation (\ref{Jlinear}) with period 2 coefficients 
\beq\label{N3Js} 
J_0 = \frac{x_2 + x_0 +a }{x_1}, \qquad J_1 = \frac{x_{1}x_0 + x_{2}x_1 + ax_0+ ax_1 +ax_2}{x_{2}x_0}.  
\eeq 
\end{example}

\section{Monodromy and the dressing chain}

Any function of the periodic quantities $J_0,J_1,\ldots ,J_{N-2}$ 
that is invariant under cyclic permutations provides a first integral 
for the $N$-dimensional map (\ref{bir}) corresponding to (\ref{myrec}). 
The $N-1$ quantities $J_i$ are independent 
functions of $x_0,x_1, \ldots , x_{N-1}$, so a choice of $N-1$ independent cyclically symmetric functions of these $J_i$ gives 
the maximum number of independent first integrals 
for such a map  in $N$ dimensions, as long as it is not purely periodic. Hence the coefficients in (\ref{Klinear}) should
be functions of the $J_i$, and writing them explicitly as such will allow us to show that $K^{(1)} = K^{(2)}$, by 
applying monodromy arguments. 

\subsection{$3\times 3$ monodromy matrices}\label{mono3}

From the linear relation (\ref{Jlinear}), the matrix $\Psi_n$ satisfies
\begin{equation} 
\Psi_{n+1} = \lax_n \Psi_n , 
\qquad \mathrm{where} \qquad 
\label{Lmat}
\lax_n =
\begin{pmatrix}
0 & 1 & 0 \\
0 & 0 & 1 \\
1 & -1-J_n & 1+J_{n+1}
\end{pmatrix}.
\end{equation}
Shifting the indices $N-2$ times  in the above linear equation for $\Psi_n$ gives the monodromy over the period 
of the coefficients, 
such that
$$
\Psi_{n+p}=\mma_n\Psi_n
\qquad \mathrm{with} \qquad 
\mma_n = \lax_{n+N-2}  \lax_{n+N-3}\cdots 
\lax_n, 
$$
where $p=N-1$ is the period.  
From the second linear relation (\ref{Klinear}), the matrix $\Psi_n$ also satisfies
\begin{equation}\label{Lhat}
\Psi_{n+p} = \Psi_n \hat{\lax}, 
\qquad \mathrm{where} \qquad  
\hat{\lax} =
\begin{pmatrix}
0 & 0 & 1 \\
1 & 0 & -K^{(2)} \\
0 & 1 & K^{(1)}
\end{pmatrix}.
\end{equation}
By Remark \ref{munon}, $\Psi_n$ is invertible, so 
we can rewrite these two monodromy equations as
$
\mma_n = \Psi_{n+p}\Psi_n^{-1} 
$
and
$
\hat{\lax} = \Psi_n^{-1}\Psi_{n+p},
$
and then taking the trace of each yields 
\begin{equation}\label{treqn}
 K^{(1)}= \tr \hat{\lax} =\tr \mma_n  \qquad \mathrm{for} \,\, \mathrm{all} \,\, n. 
\end{equation}
Thus we can write the invariant $K^{(1)}$ in terms of the periodic functions $J_i$. 
Similarly, from 
$
\mma_n^{-1} = \Psi_n \Psi_{n+p}^{-1}
$
and
$
\hat{\lax}^{-1} = \Psi_{n+p}^{-1}\Psi_n.
$
we also have
\begin{equation}\label{trinv}
 K^{(2)} = \tr \hat{\lax}^{-1} =\tr \mma_n^{-1} \qquad \mathrm{for} \,\, \mathrm{all} \,\, n.
\end{equation}
However, it is still not obvious from the form of the monodromy matrix $\mma_n$ that $\tr \mma_n=\tr \mma_n^{-1}$, 
which is required for the two non-trivial coefficients in (\ref{Klinear}) to be the same. 
Before considering the general case, we give a couple of examples. 
%
\begin{example}\label{N3a} 
For the third order case (\ref{3rdmyrec}), as in Example \ref{N3}, 
we have the monodromy matrix  
$
\mma_{0}=\lax_1 \lax_0, 
$
where the matrix in (\ref{Lmat}) has period 2: $\lax_{n+2} =\lax_n$ for all $n$. 
By taking traces it can be verified directly that $K^{(1)} = K^{(2)}=K$,  where 
\beq\label{KN3} 
K = \tr \mma_0=\tr \mma_0^{-1} =J_0J_1-1. 
\eeq 
The above identity, expressing $K$ in terms of $J_0$ and $J_1$, 
can also be checked  by comparing the Laurent polynomials 
in (\ref{constK}) and (\ref{N3Js}). 
\end{example} 
\begin{example}\label{N4} 
For $N=4$ the nonlinear recurrence (\ref{myrec}) becomes
$$
x_{n+4}x_n = x_{n+3}x_{n+1}+a(x_{n+1}+x_{n+2}+x_{n+3}).
$$
The functions defined by (\ref{Jfn}), which appear as 
entries in (\ref{Lmat}), have period 3. They are 
\begin{equation*}
J_0 = \frac{x_0+x_{2}+a}{x_{1}}, \quad J_1 = \frac{x_{1}+x_{3}+a}{x_{2}}, \quad
J_2 = \frac{x_{0}x_{2}+x_{1}x_{3} +a(x_0+x_{1}+x_{2}+x_{3})}{x_0x_{3}}, 
\end{equation*}
with $J_{n+3}=J_n$ for all $n$. 
Upon taking the trace of monodromy matrix  
$
\mma_{0}=\lax_2 \lax_1 \lax_0,  
$ 
and that of its inverse, it follows from (\ref{treqn}) and (\ref{trinv}) that $K^{(1)} = K^{(2)}=K$,  where 
\beq\label{KN4} 
K = J_0J_1J_2-(J_0+J_1+J_2)+1 
\eeq 
is the explicit formula for $K$ as a symmetric function of $J_0,J_1,J_2$. 
\end{example} 
%
\subsection{Connection with the dressing chain}\label{dcmono}

A one-dimensional Schr\"odinger operator 
can be 
factorised 
as $L=-(\partial + f)(\partial -f)$, where $\partial$ denotes the operator of differentiation with 
respect to the independent variable, $z$ say, and $f=f(z)$.  The dressing chain, as described in \cite{vesshab93}, 
is the set of ordinary differential equations
\beq\label{dress} 
(f_i+f_{i+1})' = f_i^2-f_{i+1}^2+\alpha_i, 
\eeq 
where $\alpha_i$ are constant parameters, and the prime denotes the $z$ derivative. These equations 
arise from  successive Darboux transformations 
$L_i \rightarrow L_{i+1}$ that are used to generate a sequence of Schr\"odinger operators $(L_i)$; the nature of the Darboux 
transformation is such that each operator is obtained from the previous one by  
a reordering of factors and a constant shift: 
$$ 
L_i = -(\partial + f_i)(\partial -f_i) \longrightarrow L_{i+1} = -(\partial - f_i)(\partial +f_i) +\alpha_i = -(\partial + f_{i+1})(\partial -f_{i+1}). 
$$
Of particular interest is the periodic case, where $L_{i+p}=L_i$ for all $i$. In that case, all indices in (\ref{dress})  should be read mod $p$, 
and this becomes a finite-dimensional system for the $f_i$. 
As noted in \cite{vesshab93}, the properties of this system depend sensitively on the parameters $\alpha_i$: when 
$\sum_i \alpha_i\neq 0 $ the general solutions correspond to Painlev\'e transcendents  or their higher order 
analogues (see \cite{takasaki} and references for further details), while for  $\sum_i \alpha_i = 0 $ the solutions 
are of an algebro-geometric nature, corresponding to finite gap solutions of the KdV equation; the latter connection 
goes back to results of Weiss \cite{weiss}. 

Here we are concerned with the case $\sum_i \alpha_i = 0 $ only, so following \cite{vesshab93} we use parameters $\beta_i$ such that 
$\alpha_i = \beta_i -\beta_{i+1}$ mod $p$.    
For $p$ odd there is an invertible transformation $f_i$ to new coordinates, 
that is 
\beq\label{Jf}
J_i =f_i+f_{i+1}, \qquad i=1,\ldots,p, 
\eeq 
and in \cite{vesshab93} (where the coordinates $J_i$ are denoted $g_i$) it was shown that 
in this case the periodic dressing chain is a bi-Hamiltonian integrable system, meaning that it has 
a pencil of compatible Poisson brackets together with the appropriate number of Poisson-commuting first integrals to satisfy the requirements of Liouville's theorem. 
(For $p$ odd the transformation (\ref{Jf}) is not invertible, but there is a degenerate Poisson bracket for the dressing chain such that it is an integrable system 
on a generic symplectic leaf.) A recent development   was the observation in \cite{fordyhone12}
that (with all $\beta_i=0$) a combination of these compatible brackets for the dressing chain 
coordinates $J_i$ arises by reduction from the log-canonical Poisson structure for the cluster variables in cluster algebras associated with 
affine $A$-type Dynkin quivers.   In this context, a further observation was that the  linear relations between cluster variables 
found in \cite{fordymarsh09} (see also \cite{assem,fordy10,keller}), which are of the form 
\beq\label{Alin}
x_{n+2p} -\kappa \, x_{n+p} +x_n=0, 
\eeq  
have a non-trivial coefficient $\kappa$ which is the generating function for the first integrals of the dressing chain 
(expressed  in terms of suitable variables $J_i$). The key to understanding these observations is the following 
result, which is omitted from the published version of \cite{fordyhone12} but appeared in the original preprint.  
\begin{lemma} \label{trid} 
If a $2\times 2$ monodromy matrix is defined by  
\beq\label{mstar} 
\mma^* 
=\left(\bear{cc} 0 & \zeta_1 \\ 1 & J_1 \eear\right) \, \left(\bear{cc} 0 & \zeta_2 \\ 1 & J_2 \eear\right) \ldots
\left(\bear{cc} 0 & \zeta_p \\ 1 & J_p \eear\right), 
\eeq 
then (with indices read $\bmod\,  p$) 
the trace is given by the explicit formula
$$\tr  \mma^* =
\prod_{i=1}^p \left( 1+\zeta_i\frac{\partial^2}
{\partial J_i \partial J_{i+1}}\right) \, \prod_{n=1}^p J_n. 
$$
\end{lemma}
\begin{proof} 
This is very similar to the proof of Theorem 2 in \cite{vesshab93}, 
where a Lax equation 
for the dressing chain  
is given for a different monodromy matrix, 
written in terms of the variables $f_i$.   
The above identity for the trace of $\mma^*$ follows by calculating the coefficient of each monomial in the $\zeta_j$ on either side 
of the equation, using
the relations
$$
\left(\bear{cc} 0 & 0 \\ 1 & J_1 \eear\right) \, \left(\bear{cc} 0 & 0 \\ 1 & J_2 \eear\right) \ldots
\left(\bear{cc} 0 & 0 \\ 1 & J_\ell \eear\right) = J_1J_2\ldots J_{\ell -1} \,
\left(\bear{cc} 0 & 0 \\ 1 & J_\ell \eear\right)
$$
and
$$
\left(\bear{cc} 0 & 0 \\ 1 & J_j \eear\right)\, \left(\bear{cc} 0 & 1 \\ 0 & 0 \eear\right) \,
\left(\bear{cc} 0 & 0 \\ 1 & J_k \eear\right) =
\left(\bear{cc} 0 & 0 \\ 1 & J_k \eear\right),
$$
and noting that in the latter relation the central matrix on the left is nilpotent.
\end{proof} 
\begin{remark} 
For the period $p$  dressing chain (\ref{dress}) with $\alpha_i=\beta_i - \beta_{i+1}$, the Poisson-commuting first integrals  
are found in terms of the variables $J_i$  by setting $\zeta_i=\beta_i - \lambda$ in (\ref{mstar}) and 
expanding   $\tr  \mma^*$ in powers of the spectral parameter $\lambda$. 
\end{remark}
The $2\times 2$ monodromy matrix (\ref{mstar}) also provides the key to understanding the detailed structure of 
the linear relation (\ref{Klin}) associated with (\ref{myrec}). 
\begin{theorem}\label{2mdy} 
In terms of the quantities $J_n$ defined in (\ref{Jfn}), the coefficient $K$ in (\ref{Klin})  is given  by the formula 
\beq\label{kform} 
K= 1+ 
\prod_{i=1}^p \left( 1 - \frac{\partial^2}
{\partial J_i \partial J_{i+1}}\right) \, \prod_{n=1}^p J_n. 
\eeq  
\end{theorem} 
\begin{proof}
Upon introducing the $2\times 2$ matrices 
$$ 
\Phi_n = \left(\bear{cc} x_n & x_{n+1} \\ x_{n+p} & x_{n+p+1} \eear\right), \quad 
\lax^*_n = \left(\bear{cc} 0 & -1 \\ 1 & J_1 \eear\right), \quad 
\cax = \left(\bear{cc} 0 & 1 \\ 0 & 1  \eear\right), 
$$ 
the inhomogeneous linear equation (\ref{Jarrange}) implies that 
\beq\label{phin} 
\Phi_{n+1} = \Phi_n \lax^*_n - a \cax. 
\eeq 
Seeking an appropriate inhomogeneous counterpart of   (\ref{Klin}), we define the 
matrix 
$$ 
\mma^*_n = \lax^*_{n}  \lax^*_{n+1}\cdots 
\lax^*_{n+p-1}, 
$$
which (up to cyclic permutations) is a special case of the monodromy matrix $\mma^*$ above,  obtained  from (\ref{mstar}) by setting $\zeta_i = -1$ for all $i$. 
Now, by applying a shift to (\ref{phin}), we have 
$$  
\Phi_{n+2} =(\Phi_n\lax^*_n - a\cax)\lax^*_{n+1}-a\cax = \Phi_n\lax^*_n \lax^*_{n+1} - a\cax (\lax^*_{n+1} + {\bf I}) 
$$ 
(with ${\bf I}$ being the $2\times 2$ identity matrix), 
and so  by induction, after a total of $p-1$ shifts,  we find that 
\beq\label{pshi} 
\Phi_{n+p} = \Phi_n \mma^*_n -  a\cax\kax_n, 
\eeq  
where 
$$ 
\kax_n = \lax^*_{n+1}\cdots\lax^*_{n+p-1}+\lax^*_{n+2}\cdots\lax^*_{n+p-1}+\ldots + \lax^*_{n+p-1} + {\bf I}. 
$$

Note that the entries of $\lax^*_n$, and hence those of $\mma^*_n$ and $\kax_n$, are all periodic with period $p$, 
so from (\ref{pshi}) it follows that 
$$ 
\Phi_{n+2p} = \Phi_{n+p}\mma^*_n -  a\cax\kax_n = \Phi_n (\mma^*_n)^2  -a\cax\kax_n(\mma^*_n+{\bf I}). 
$$ 
The latter equation can be further simplified with the use of the Cayley-Hamilton theorem, which 
gives 
$$ 
(\mma^*_n)^2 - \kappa \mma^*_n+ {\bf I}=0, \qquad \mathrm{with} \quad \kappa = \tr \mma^*_n, 
$$ 
and hence (after using (\ref{pshi}) once again) we find 
$$
\Phi_{n+2p} = \kappa \Phi_{n+p} - \Phi_{n} - a\cax\kax_n\big(\mma^*_n+(1-\kappa){\bf I} \big). 
$$ 
The $(1,1)$ entry of the above matrix equation yields an inhomogeneous version of (\ref{Alin}), namely  
\beq\label{inhgk} 
x_{n+2p}-\kappa x_{n+p} + x_n +a \tilde{J}_n=0, 
\eeq 
where the quantity  $\tilde{J}_n$ is periodic with period $p$. 
We claim that this is the desired   inhomogeneous counterpart of   (\ref{Klin}).  Indeed, applying the shift $n\to n+p$ to (\ref{inhgk}) 
and then subtracting the original equation gives 
$$ 
x_{n+3p}  -(\kappa +1) \big(x_{n+2p}- x_{n+p}\big) -x_n = 0, 
$$  
which is (\ref{Klin}) with $K=\kappa + 1$, and the formula (\ref{kform}) then follows from Lemma \ref{trid}. 
This also verifies that $K^{(1)} = K^{(2)}$ in (\ref{Klinear}), completing the missing step in the proof of 
Theorem  \ref{thmlin}. 
\end{proof} 
\begin{remark} 
 Observe that 1 is a root of the characteristic polynomial of the  linear relation  (\ref{Klin}), which means that this homogeneous  
equation is a total difference.  
 Hence, using $\shi$ to denote the shift operator, we obtain the relation 
$$ 
\big( 
\shi^{3p-1} + \ldots +\shi^{2p} -\kappa (\shi^{2p-1} + \ldots +\shi^{p}) + \shi^{p-1} + \ldots +\shi + 1 
\big) \, x_n +a\tilde{K}=0, 
$$ 
where $\tilde{K}$ is a first integral. In the case $p=N-1=2$, a formula for   $\tilde{K}$ 
as a Laurent polynomial in the initial data can be found in \cite{honechapt} (see Theorem 5.1 therein). 
Although we do not have a general formula for $\tilde{K}$, analogous to (\ref{kform}), 
direct calculation shows that $\tilde{K}=J_0+J_1+4$ for $p=2$, and 
$
\tilde{K}=J_0J_1+J_1J_2+J_2J_0 +2 ( J_0+J_1+J_2) + 3
$
for $p=3$. Moreover, in the cases $p=2,3,4,5$ we have verified with 
computer algebra that $\tilde{K}$ is given in terms of the 
$3\times 3$ determinant (\ref{mu}) by 
$$ 
\tilde{K}=\frac{\mu}{a^3}, 
$$ 
and we conjecture that this is always so.  
\end{remark} 
\begin{remark} 
The form of the inhomogeneous relation (\ref{inhgk}) means that 
the solution of the initial value problem for (\ref{myrec}) can 
be written explicitly using Chebyshev polynomials of the first and 
second kind with argument $\kappa /2$,  that is    
$$ 
T_j(\cos \theta) = \cos n\theta , \qquad 
U_j(\cos \theta) = \frac{\sin (n+1)\theta}{\sin \theta} ,  
\qquad \mathrm{where} \quad  \kappa =2\cos\theta. 
$$
For details of the solution in the case $N=3$, see Proposition 5.2 in    \cite{honechapt}. 
\end{remark} 
\section{Conclusions}

We have shown that each of the nonlinear recurrences (\ref{myrec}) fits into the framework of Lam and Pylyavskyy's LP algebras, implying that the 
Laurent property holds. Furthermore, each recurrence is linearisable in two different ways: with a constant-coefficient homogeneous linear relation, and with another such 
relation with periodic coefficients. The connection between the first integral $K$ in (\ref{Klin}) and the Poisson-commuting first integrals of the dressing chain is intriguing:  
it raises the question of whether LP algebras might admit natural Poisson and/or presymplectic structures, analogous to those for 
cluster algebras \cite{gsvduke}.

The recurrences (\ref{hhodd}) considered by Heideman and Hogan also possess the Laurent property. (See \cite{hoganthesis} for a more detailed discussion.) It can also be shown 
that, beyond the linear relation  (\ref{heholin})  for a particular integer sequence,  the result of Proposition \ref{hhthm} extends to arbitrary initial data.  In this case there is also an analogue of (\ref{Jlinear}), 
with periodic coefficients, but with a more complicated structure. These results will be presented elsewhere.

\section{Appendix} 

Here we consider the particular family of integer sequences that is generated by 
iterating each recurrence (\ref{myrec}) 
with all  $N$ initial data set to the value 1, and prove 
the following result. 
\begin{theorem}\label{init1s} 
For each $N\geq 2$, the recurrence (\ref{myrec}) with initial values 
$x_n=1$ for $n=0,\ldots, p=N-1$ generates a sequence of  integers $(x_n)$  
which is symmetric in the sense that $x_n=x_{p-n}$ for all $n\in\mathbb{Z}$. 
The first few terms are given by 
\beq\label{terms} 
x_{p+j}=\fib_{2j}\, p+1, \quad  
x_{2p+j}=\fib_{2j}\fib_{2p}\, p^2 +\Big(\fib_{2j}(\fib_{2p+2}+1)-\fib_{2j-2}\fib_{2p}\Big) p+1, 
\quad j=0,\ldots , p. 
\eeq 
In general, for each $k\geq 0$ the terms $x_{kp+j}$ for $1\leq j\leq p$ can be 
written as polynomials in $p$ whose coefficients are themselves polynomials in Fibonacci numbers with even 
index.  
\end{theorem} 
First of all, to see why $x_n=x_{p-n}$, observe that the nonlinear 
recurrence  (\ref{myrec}) has two obvious symmetries: translating $n$, and replacing $n$ by $-n$. 
Combining these two symmetries, it follows that the sequence $(x_{p-n})$ 
satisfies (\ref{myrec}) whenever the sequence $(x_n)$  does; 
the fact that the 
initial data 
$1,1,\ldots ,1$ is symmetric under $n\to p-n$ 
implies 
that these two sequences 
coincide. 

The rest of the proof 
proceeds inductively, and relies on various 
identities for Fibonacci numbers; note that 
the convention $\fib_0=0$, $\fib_1 =1$ is taken for the initial Fibonacci numbers.  
From the Fibonacci recurrence $\fib_{n+2}=\fib_{n+1}+\fib_n$ 
it follows that the even index terms satisfy 
\beq\label{fibdo} 
\fib_{2n+4}=3\fib_{2n+2}-\fib_{2n}.  
\eeq 
For $n=0,\ldots , p-1$  
the relation (\ref{myrec}) is satisfied by the initial data and the first set of terms in (\ref{terms}) by virtue of 
the identity 
\beq \label{fibid}
\fib_{2n+2}-\fib_{2n} = \fib_{2n}+\fib_{2n-2}+\ldots + \fib_2 + 1,   
\eeq  
which follows from (\ref{fibdo}) by induction. To verify the formula for the second set of terms  in (\ref{terms}),  
it suffices to substitute these expressions into (\ref{myrec}) with $n=p,\ldots , 2p-1$ 
and compare coefficients of powers of $p$ on both sides, making use of (\ref{fibid}) again, together 
with the identity 
$$ 
\left|\begin{array}{cc} \fib_{2j-2} & \fib_{2j} \\ 
           \fib_{2j} &     \fib_{2j+2} 
\end{array}\right| =-1. 
$$ 
The fact that the latter $2\times 2$ determinant is a constant (independent of $j$) follows from  (\ref{fibdo}), 
while the value $-1$ is obtained immediately from  $\fib_0=0$, $\fib_{\pm 2}=\pm 1$.           

Proposition \ref{inits}  is now seen to be a special case of Theorem \ref{thmlin},  by making use of 
Theorem \ref{init1s}. Indeed, from (\ref{Klin}) with $n=j-p$ for $1\leq j\leq p$, the above choice of initial data 
and the terms in (\ref{terms}), as well as the symmetry of the sequence, yields 
$K=(x_{2p+j}-x_{-p+j})/(x_{p+j}-x_j)=(x_{2p+j}-x_{2p-j})/(x_{p+j}-x_j)=\fib_{2p} p+\fib_{2p+2}+1-(\fib_{2j-2}\fib_{2p}+\fib_{2p-2j})/\fib_{2j}$. 
In order to see why the latter expression is a constant (independent of $j$), and to obtain the 
particular form of the coefficient 
$K$ that  appears in (\ref{inik}), it is enough to verify the identity 
$\fib_{2j}\fib_{2p-2} - \fib_{2j-2}\fib_{2p}= \fib_{2p-2j}$, 
which is 
done by writing the even index Fibonacci numbers as 
$ \fib_{2n} = \mathrm{sinh}\, n\theta / \mathrm{sinh}\, \theta$ with $\theta = \log \left(\frac{3+\sqrt{5}}{2}\right) $, 
and using the addition formula for the hyperbolic sine. 


\affiliationone{
   A.N.W. Hone  and C. Ward \\
   School of Mathematics, Statistics \& Actuarial Science\\ 
   University of Kent, 
   Canterbury CT2 7NF\\ 
   U.K.
 \email{anwh@kent.ac.uk\\ 
            cw336@kent.ac.uk}}
\end{document}